\documentclass[runningheads]{llncs}
\usepackage[T1]{fontenc}
\usepackage{hyperref}

\usepackage{color}

\usepackage{booktabs}       
\usepackage{amsfonts}       
\usepackage{nicefrac}       
\usepackage{microtype}      

\usepackage{subcaption}

\usepackage{breakcites} 

\usepackage{appendix}

\usepackage{hyphenat}
\usepackage{xspace}

\newcommand{\xhdr}[1]{\vspace{1.7mm}\noindent{{\bf #1.}}}

\newcommand{\ie}{{i.e.}\xspace}

\newcommand{\cf}{{cf.}\xspace}

\usepackage{bm}

\newcommand{\Secref}[1]{Sec.~\ref{#1}}

\newcommand{\Thmref}[1]{Thm.~\ref{#1}}

\newcommand{\Lemmaref}[1]{Lemma~\ref{#1}}

\newcommand{\Defref}[1]{Def.~\ref{#1}}

\newcommand{\Figref}[1]{Fig.~\ref{#1}}

\newcommand{\Appref}[1]{Appendix~\ref{#1}}

\usepackage{bbm} 
\usepackage{mathtools}
\DeclarePairedDelimiter\abs{\lvert}{\rvert}
\DeclarePairedDelimiter\norm{\lVert}{\rVert}

\usepackage{amsmath}

\usepackage[normalem]{ulem} 

\newcommand{\R}{\mathbb{R}}

\newcommand{\Z}{\mathbb{Z}}
\newcommand{\defeq}{\mathrel{\vcentcolon=}}

\DeclareMathOperator{\RNG}{RNG}

\newcommand{\leftmsgsmall}[1]{\>\hspace{-2cm}\sendmessageleft*{#1}\hspace{-2.5cm}\>}
\newcommand{\rightmsgsmall}[1]{\>\hspace{-2.5cm}\sendmessageright*{#1}\hspace{-2cm}\>}

\usepackage{cryptocode}
\usepackage{wrapfig}

\newcommand*{\msetl}{\{\mskip-4mu\{}
\newcommand*{\msetr}{\}\mskip-4mu\}}

\newcommand\restr[2]{{
  \left.\kern-\nulldelimiterspace 
  #1 
  \right|_{#2} 
  }}
  
\usepackage[rightcaption]{sidecap}

\usepackage{thmtools}
\usepackage{thm-restate}

\begin{document}
\title{Secure Summation via Subset Sums:\\A New Primitive for Privacy-Preserving Distributed Machine Learning}
\titlerunning{Secure Summation via Subset Sums}
%
%
%
\author{Valentin Hartmann \and
Robert West}
\authorrunning{V. Hartmann et al.}
%
\institute{EPFL\\
\email{\{valentin.hartmann,robert.west\}@epfl.ch}}

\maketitle              
\begin{abstract}
For population studies or for the training of complex machine learning models, it is often required to gather data from different actors. In these applications, summation is an important primitive: for computing means, counts or mini\hyp batch gradients. In many cases, the data is privacy\hyp sensitive and therefore cannot be collected on a central server. Hence the summation needs to be performed in a distributed and privacy\hyp preserving way. Existing solutions for distributed summation with computational privacy guarantees make trust or connection assumptions --- e.g., the existence of a trusted server or peer\hyp to\hyp peer connections between clients --- that might not be fulfilled in real world settings. Motivated by these challenges, we propose Secure Summation via Subset Sums (S5), a method for distributed summation that works in the presence of a malicious server and only two honest clients, and without the need for peer\hyp to\hyp peer connections between clients. S5 adds zero\hyp sum noise to clients' messages and shuffles them before sending them to the aggregating server. Our main contribution is a proof that this scheme yields a computational privacy guarantee based on the multidimensional subset sum problem. Our analysis of this problem may be of independent interest for other privacy and cryptography applications.

\end{abstract}
\section{Introduction}
Summation and averaging are primitives used in virtually every analysis of data. With the rise of neural networks, averaging of gradients has also become an essential part of machine learning model training.
The amount of data that is collected to compute these sums and averages is increasing, and the data is collected in ever more places: from phones, smartwatches, Internet-of-things devices, cars, etc.
At first glance, it seems as though this abundance of data should satisfy the needs of statistical analyses and data\hyp hungry machine learning models --- if only we could pool the data from all the different places. The caveat is that much modern data is privacy\hyp sensitive.

Hence, while the rate at which data is being produced is ever\hyp increasing, data collection and pooling has in a certain sense become not easier but harder, both because of new data protection laws such as the GDPR \cite{eu:gdpr} and because of more awareness in the wake of data scandals such as those surrounding Cambridge Analytica \cite{cambridge2018} or Strava~\cite{strava2018}.

\xhdr{Previous work}
A straightforward approach to mitigating privacy concerns would be to add noise to datasets and then merge them. However, this may require so much noise that it would render the data essentially useless \cite{duchi2018minimax,fienberg2010differential}.
Nowadays' research thus focuses on privacy\hyp preserving protocols that can compute sums on distributed datasets. They can be roughly divided into two groups: protocols based on differential privacy that provide privacy at the cost of summation accuracy, e.g., \cite{duchi2018minimax}; and protocols that do not come with a loss in accuracy, but instead with an increase in communication cost, e.g., \cite{bonawitz2017practical}. A main disadvantage of the latter approaches is that they typically assume peer\hyp to\hyp peer connections between clients or make trust assumptions on the clients or, for example, the existence of a trusted server for the distribution of cryptographic keys. See \Secref{sec:related} for a detailed discussion of related work.

\xhdr{Contributions: S5} In this paper we propose Secure Summation via Subset Sums (S5), a method that allows for computing sums of vectors in a setting with a central server and distributed data. It belongs to the group of approaches that do not compromise summation accuracy but instead increase communication to provide a computational privacy guarantee. Our guarantee holds both against a malicious server and against malicious clients. As opposed to previous methods, S5 does not require peer\hyp to\hyp peer connections between clients, and the only trust assumptions are the existence of two honest clients and a mix network (mixnet) \cite{chaum1981untraceable} such as Tor \cite{syverson2004tor}. The main ideas of S5 are not new; our contribution is their specific combination and the proof that this combination yields a computational privacy guarantee. S5 lets clients add noise vectors to their summands prior to sending them, in a way that removes all information about the original summands, thereby making the summands useless for the extraction of sensitive information about a client. The summands' usefulness for downstream tasks is restored in a second step, where the negative value of the noise is sent to the server through a mixnet. Adding this negative noise to the sum of the noisy summands yields the sum of the original summands. However, no single local summand can be recovered, because in the mixnet the noise vectors of all clients are shuffled. This turns the task of linking a noise vector to a specific noisy summand into a computationally infeasible instance of the subset sum problem, which turns the task of breaking privacy into a computationally hard problem.
\Figref{fig:protocol} provides an overview of S5.

\xhdr{Properties of S5}
During the summation process, the server learns nothing but the sum of the clients' summands, where each client contributes a single summand. These are the key properties of S5:
\begin{itemize}
    \item Main idea: privacy by adding and canceling noise.
    \item Computational privacy guarantee: compromising privacy would require solving a hard instance of the multidimensional subset sum problem, which we show to be computationally intractable.
    \item Trust assumptions: at least two honest clients; a mixnet such as Tor \cite{syverson2004tor}. The aggregating server may be malicious.
    \item Increase of communication compared to standard distributed summation by only a logarithmic factor in the number of users and the dimensionality of the summands.
    \item No need for peer\hyp to\hyp peer connections between clients.
    \item To the best of our knowledge, the first method without loss in summation accuracy that works under these trust and connection assumptions and provides a computational privacy guarantee.
    \item Limitation to \emph{either} small vectors and few clients, \emph{or} settings where clients have a stable Internet connection and are unlikely to drop out during the summation process.
\end{itemize}

Despite the limitation mentioned last, S5 applies to many situations of practical importance, e.g., 
multiple hospitals that each collect patient data and want to pool the data across hospitals to answer research questions;
a franchise company that wants to analyze customer behavior via purchase logs in its different branches; 
or a vendor of server software that wants to improve its product's performance by optimizing it for common usage patterns. In this last setting, the Internet connections would be stable because the software would run on servers. However, the software would continuously get new users and lose old ones, and the users typically would not know each other. Hence distributing cryptographic keys, which are required for many existing protocols for privacy\hyp preserving summation, manually between the different users would not be an option.

\xhdr{Application to federated learning}
S5 can be used to train machine learning models on distributed data in the federated learning (FL) paradigm \cite{federatedlearning}. In FL, clients compute model updates w.r.t.\ their local datasets and send them to a server, which uses the average of these updates to update the model. S5 can be used for the averaging to prevent the server from extracting information about the datasets from the model updates. We discuss this application in more detail in \Secref{sec:fl}.

\xhdr{Organization of the paper}
We begin by discussing related approaches and how they differ from S5 in \Secref{sec:related}. We then define the problem and the adversarial model in \Secref{sec:problem}, and present our solution in \Secref{sec:solution}. \Secref{sec:guarantee} contains the technical details of our privacy guarantee. In \Secref{sec:experiments} we present the results of various experiments with an implementation of S5. In \Secref{sec:fl} we describe how S5 can be used for federated learning. \Secref{sec:limitations} contains a discussion of the limitations of our work. A summary of the paper is given in \Secref{sec:conclusion}.

\section{Related Work}
\label{sec:related}
There are two general approaches towards protecting privacy when computing sums of distributed vectors: by perturbing the summands or by using cryptographic approaches.

The former approach leads to local differential privacy guarantees by adding independent noise to the summands before sending them to the server \cite{dwork2006calibrating,kasiviswanathan2011can}. An advantage of this approach is that in addition to the summands, the sum itself is also privacy protected. That is, the sum cannot be used to extract private user information because it is noisy. However, in this standard form, large amounts of noise need to be added, since each client's data needs to be protected individually without being able to rely on other clients \cite{duchi2018minimax,wang2019collecting}.
The addition of noise that cancels out in combination with a shuffle mechnanism, very much as in our method, has been used to prove differential privacy guarantees \cite{balle2020private,ghazi2019scalable}. However, in order to obtain a differential privacy guarantee, additionally a certain amount of noise that does not cancel out and thus reduces utility needs to be added. Our method can do without such noise, because we give a computational privacy guarantee instead.

Cryptographic approaches based on homomorphic encryption \cite{aono2016privacy,aono2018privacy} typically try to protect a private client's data from an aggregating server but not from other clients. Clients encrypt their summands with the same key before sending them to the server, which then performs the summation in the encrypted space. The resulting sum can be decrypted by the clients. The encryption scheme of
Shi et al.
\cite{shi2011privacy} allows the clients to encrypt their summands with different keys, but requires a trusted setup phase with a trusted server, or communication between the clients. Methods based on generic secure multiparty computation use secret sharing \cite{beimel2011secret}, where a secret value is distributed between the clients in several parts such that a certain number of clients is required to reconstruct the secret value. These protocols require communication between the clients and often a large number of honest clients \cite{ben1988completeness,chaum1988multiparty}. Methods without this last restriction \cite{lindell2011ips} still need direct communication between clients, which might not be possible, e.g., in the setting from the introduction where the vendor of a server software wants to gather usage statistics, since especially database servers often only allow connections from a whitelist of IP addresses for security reasons.
Bonawitz et al.
\cite{bonawitz2017practical} propose a protocol that is based on a similar idea as ours, i.e., adding noise that cancels out later. However, it requires a trusted server to distribute keys in the setup phase. Trusting an already established distributed public infrastructure (e.g., Tor), as in our case, typically comes with a much lower risk than having to trust a single party that sets up and manages a server with software that is used for only one specific application.

\section{Problem Definition}
\label{sec:problem}
In our setting, there are \(N\) clients, each client \(i\) with a share \(D_i\) of a dataset \(D\). Think, e.g., of different hospitals, each with a database with information about their patients. Clients are connected via the Internet to a (potentially malicious) server that wants to compute the sum \(s = \sum_{i=1}^N s_i\) over the values \(s_i = f_p(D_i)\) of a function \(f_p\) computed on the individual datasets \(D_i\). For example, \(f_p\) could simply compute the mean of one column of \(D_i\), or it could be the gradient function of a machine learning model. \(f_p\) takes values in \(\R^d\) and is parametrized by a vector \(p\).

\xhdr{Goal and adversarial model}
The server should learn the sum \(s\). At the same time, we want to prevent it from learning any of the summands \(s_i\) (if they cannot be learned from \(s\) itself), as they might contain sensitive information about a client's dataset. We assume that the server is actively malicious: in order to break privacy and to learn any of the \(s_i\), it may deviate from the summation protocol. In addition, the server may collude with any but two clients, and all of the colluding clients may deviate from the protocol as well.

\section{Proposed Solution: Secure Summation via Subset Sums (S5)}
\label{sec:solution}
A fact also used in other privacy protocols \cite{bonawitz2017practical} is that when computing a sum, it does not matter if one adds additional summands that sum up to 0. Instead of adding a 0 sum directly, our method operates in two steps: first, random vectors are added to the summands to obfuscate them, then the same vectors are subtracted again. While many of the elements of our method are---at least individually---not new (zero\hyp sum noise, shuffling, sending seeds for a random number generator), we are the first to combine them into a protocol with minimal connection and trust assumptions, and prove that this combination provides a computational privacy guarantee based on the subset sum problem.

\subsection{Preprocessing}
\label{sec:preprocessing}
For our solution we need to represent the entries of the summands \(s_i\) as elements from the group \(\Z_{2^m}\), \ie, the integers modulo \(2^m\) for an integer \(m\). Here, \(2^m\) is an upper bound on the entries of the sum \(s\), derived from an upper bound \(2^{\tilde{m}}\) on the entries of the summands \(s_i\) and the number \(N\) of users: \(m=\lceil\log(N)\rceil + \tilde{m}\), since when summing \(N\) values, at most \(\lceil\log(N)\rceil\) additional carry bits are needed. If the summands \(s_i\) are real\hyp valued, the clients can transform them to elements from \(\Z_{2^m}\) in a preprocessing step that we detail in the \Appref{sec:app_preprocessing}.

\subsection{Protocol}
\xhdr{Main idea} The idea behind S5 is to first add noise to the summands prior to sending them, let the server sum them up to obtain a noisy version of \(s\), and to then tell the server how much noise was added so that it can remove the noise from \(s\). However, the noise vectors of the different clients get shuffled throughout the process, turning the reconstruction of any of the \(s_i\) into a computationally hard subset sum problem \cite{kleinberg06}. So instead of sending one message containing \(s_i\), each client \(i\) additionally generates \(K\) independent random vectors \(r_{i1},\dots,r_{iK}\) and sends the following \(K+1\) messages:
\begin{equation*}
    (1)\ \tilde{s}_i \defeq s_i + \sum_{k=1}^K r_{ik}, \hspace{5mm}
    (2)\ -r_{i1},
    \hspace{5mm}
    \dots,
    \hspace{5mm}
    (K+1)\ -r_{iK}.
\end{equation*}
To obtain \(s\), the server simply has to sum up all messages it received, so from the utility perspective nothing has changed over summing the \(s_i\) directly. What about privacy? If the server knows which of the messages were sent by the same client \(i\), summing them up reveals \(s_i\)---exactly what we want to avoid.
The server has two ways to link messages from the same client with each other: (1) via their content and (2) via the metadata of the network packets. We will discuss both of those in the following two paragraphs.

\xhdr{Packet content} For making the messages unlinkable via the vectors they contain, we need to make them, or at least the \(r_{ik}\), all look indistinguishable. This can easily be done by sampling them
independently
from the same distribution. We, however, also do not want \(\tilde{s}_i\) itself to carry any information about \(s_i\). This could for example happen if \(K\) were small and the \(r_{ik}\) were sampled from a distribution with small variance. This is why we choose the uniform distribution on \(\Z_{2^m}^d\) for the noise vectors, that is,
\(r_{ik} {\sim} \mathcal{U}(\Z_{2^m}^d)\) i.i.d.
As a consequence, \(\tilde{s}_i\), too, is uniformly distributed on \(\Z_{2^m}^d\). Furthermore, any \(K\)\hyp element subset of the \(K+1\) messages a client sends is statistically independent. In \Secref{sec:guarantee} we show that the information that still remains in the set of messages cannot be used by a computationally bounded adversary if we choose \(K=dm/2\).

\xhdr{Metadata} There are two types of metadata that the server receives from each packet, which the server could use to link them to clients:
(1)~the source IP address and
(2)~the arrival time, which can be used to guess the sending time.
The IP address can be removed by routing the messages through different machines, which is a functionality provided by, e.g., a mix net such as the Tor network \cite{syverson2004tor}. Further, the server can be prevented from gaining information from the packet arrival times by letting the clients send all of their messages at random times within the same interval \cite{hartmann2019privacy}.

\xhdr{Malicious adversary} So far we worked under the assumption that the server is honest but curious, i.e., that it might try to infer additional information from the data it receives, but that it at least honestly follows the protocol. An actively malicious adversary, however, might, e.g., tell only a single client \(i\) to send its vectors during a specific time period, and could thereby reconstruct \(s_i\) by summing up all vectors received during this time period. We can remove the assumption of an honest but curious adversary by either hard\hyp coding the information necessary for executing the protocol in the client software, or by letting clients request this information multiple times from the server and only send their data if the information is the same each time, as proposed by
Hartmann et al.
\cite{hartmann2019privacy}.

\subsection{Improving Communication Efficiency}
\label{sec:efficiency}
It is not necessary to send the \(r_{ik}\) as vectors, which would be of the same, potentially high, dimension as the summands \(s_i\). Instead, the server and the clients can agree on a common random number generator (RNG) beforehand, e.g., by hardcoding it. A client then generates \(K\) seeds \(R_{i1}, \dots, R_{iK}\) and uses the RNG to compute \(r_{i1}, \dots, r_{iK}\). It then sends the vector \(\tilde{s}_i\) and the scalars \(R_{i1}, \dots, R_{iK}\), which are used by the server to compute \(r_{i1}, \dots, r_{iK}\) once again.

How many bits do we need for the seeds? Using seeds with fewer bits than the random vectors that are generated from them increases the probability of collisions, i.e., two users generating the same seeds and hence the same random vectors by chance. This might weaken the hardness guarantee. As we will see later, only collisions between the vectors of two of the users are to be avoided. If \(b\) is the number of bits used for the seeds, we can easily upper bound the collision probability \(q\) by assuming that the event of the collision of any two seeds is independent of the event of the collision of any two other seeds. We can then arrange the seeds in a list and compute the probability that the second seed collides with the first one, the probability that the third seed collides with the first or the second one and so on. Summing up yields \(q \leq \frac{2K(2K-1)}{2}\frac{1}{2^b}\).
For a desired target probability \(q\), we need to choose \(b = \log\left(\frac{2K(2K-1)}{2q}\right)\).
For \(d=10^6\) dimensions, an encoding length of \(m=30\) bits, a collision probability of \(q=10^{-10}\) and the most secure choice for \(K\), namely \(K=dm/2\) (\cf\ \Secref{sec:guarantee}), only 82 bits are required per seed.

\subsection{Computation and Communication Cost}
\label{sec:cost_analysis}
For the runtime and communication analysis we will assume that the computation of the summands \(s_i\) takes time in the order of the bit length, \(\mathcal{O}(d\tilde{m})\), where \(\tilde{m}\) is the number of bits used to represent each entry of the summands. \(d\tilde{m}\) is the size of the summands and hence a lower bound; a higher computation time would favor our protocol because it would reduce the multiplicative overhead of our method. We further assume that the query vector \(p\) is of the same dimensionality and encoding length as the summands, which is typically the case for, e.g., gradient descent, where \(p\) is the vector of model weights.

Furthermore, we have seen in \Secref{sec:efficiency} that the seeds should be represented using \(\mathcal{O}(\log(K))\) bits, and in \Secref{sec:hardness} we show that the most secure choice for \(K\) is \(K=dm/2\).

\xhdr{Runtime}
In the baseline case where each client sends their gradient directly, the runtime complexity for a client is \(\mathcal{O}(d\tilde{m})\), for the server it is \(\mathcal{O}(Nd\tilde{m})\).
For S5, every client has to sample \(K\) random seeds, generate the corresponding \(d\)\hyp dimensional random \(m\)\hyp bit vectors and add them to their summand. This has complexity \(\mathcal{O}(d^2 m^2)\). The server also has to generate those random vectors and add them up, leading to a complexity of \(\mathcal{O}(N d^2 m^2) = \mathcal{O}(N\log(N)^2 d^2 \tilde{m}^2)\).
Hence, the computation increases by a factor of \(\mathcal{O}(d\tilde{m}\log(N)^2)\) over the non\hyp private baseline, both for the clients and the server.

\xhdr{Communication}
Without S5, each client needs to request the query vector \(p\) from the server and send the summand, which are \(2d\tilde{m}\) bits in total.
Moving to S5, each client needs to receive the parameter vector and send the sum of \(s_i\) and the random vectors, together with the seeds. These are \(\mathcal{O}(d\tilde{m} + dm + K\log(K)) = \mathcal{O}(dm\log(d m)) = \mathcal{O}(d\tilde{m}\log(N)\log(d \tilde{m}\log(N)))\) bits.
When using S5, the communication of the clients and of the server therefore increases by a factor of \(\mathcal{O}(\log(N)\log(d\tilde{m}\log(N)))\) over non\hyp private summation.

\subsection{S5 in a Nutshell}
\label{sec:nutshell}
S5 operates in the group \(\Z_{2^m}^d\), to which real valued vectors must be mapped. The server first sends the query parameter vector \(p\) to the clients, which then compute the summands \(s_i\) w.r.t.\ their local dataset. They sample \(K=dm/2\) random seeds, use them to generate \(K\) uniformly random vectors in \(\Z_{2^m}^d\) and add them to their summand to obtain \(\tilde{s}_i\). Then they send \(\tilde{s}_i\) and the random seeds to the server through a mixnet. The server uses the seeds to generate the corresponding random vectors, and subtracts them from the vectors \(\tilde{s}_i\) it has received, obtaining the sum \(s\).
See \Figref{fig:protocol} for a graphical overview of S5.

\section{Privacy Guarantee}
\label{sec:guarantee}
The privacy guarantee we give is a guarantee w.r.t. the summands \(s_i\) (and not w.r.t.\ the sum \(s\)). In the following we use \(h\) to denote the summand of a client that the server is interested in and that might or might not have participated in the summation process.

The privacy guarantee says that it is impossible for the server to gain any knowledge---apart from what can be learned from the sum alone---about whether a certain client took part in the summation process or not, given full knowledge about the client's summand \(h\). If this is impossible even in the setting of full knowledge of \(h\), then the server will also not be able to learn anything about a specific client given only incomplete knowledge. In particular, it will not be able to reconstruct a client's summand from partial knowledge about the summand.

\subsection{Overview}
\label{sec:privacy_overview}
\begin{theorem}
\label{thm:main}
Assume that there exist at least two honest clients and that the set of subset sum problems with \(dm\) uniformly distributed \(d\)\hyp dimensional vectors with encoding length \(m\) per entry and the sum consisting of \(dm/2\) summands is computationally hard (\cf\ \Secref{sec:equivalence}).
Then a polynomially computationally bounded server is not able to prove that, with positive probability, for a given vector \(h\) there exists a client \(i\) among the set of honest clients such that \(h=s_i\), if this cannot be learned from \(s\) alone.
\end{theorem}
We say ``with positive probability'' because the server will not necessarily be able to definitely prove, even with unlimited computational power, that there exists a client that had a specific summand, but only whether this is possible or not, i.e., whether the answer to the decision problem from \Secref{sec:equivalence} is true or false. We further point out that the privacy guarantee against the server automatically yields the same privacy guarantee against other clients since we allow the server to collude with all but two clients.

\xhdr{Intuition} \Thmref{thm:main} is a strong guarantee: If the server is not even able to determine whether a specific client participated in the summation, then it certainly will not be able to infer any other kind of information (e.g., a part of their dataset) about a client from the data it receives. This holds even in the case when the server has arbitrary side information about the client to be attacked, e.g., when the server knows some entries of a summand and wants to infer the rest of its entries. For example, assume that each dataset contains information about one person and that one of the pieces of information in the dataset is the age of the person, another one whether they have cancer. Then the server will not be able to tell whether one of the persons has a given age \(z\), and hence definitely not whether a person of age \(z\) (that is known to the server) has cancer or not.

In the remainder of this section we first prove \Thmref{thm:main} by showing that the task of detecting a given summand is equivalent to solving a certain instance of subset sum (\Secref{sec:equivalence}). We then show that these instances are computationally hard (\Secref{sec:hardness}), and finally discuss the computational complexity of existing algorithms for subset sum problems (\Secref{sec:hardness_practice}).

In the proof of \Thmref{thm:main} we allow the server to know which of the messages sent by the users are of the type \(\tilde{s}_i\) and which are of the type \(-r_{ik}\). This is the case when using the more efficient way of sending the \(-r_{ik}\) as random seeds instead of vectors (see \Secref{sec:efficiency}). Since the \(-r_{ik}\) arrive in a random order and are all independent and identically distributed, from the server's perspective they just form one big multiset of messages, and the \(\tilde{s}_i\) form another multiset, which we will together denote by \(M=(\msetl a_1,\dots,a_N\msetr, \msetl b_1,\dots,b_{NK}\msetr)\), where we use the notation \(\msetl\cdot\msetr\) for multisets.
The complexity guarantee we give is based on the multi\hyp dimensional subset sum problem \cite{emiris2017approximating,impagliazzo1996efficient}:
\begin{definition}[\(d\)\hyp dimensional Decisional Subset Sum Problem, \(d\)\hyp SSS]
Given \(n\) vectors \(V=\msetl v_1, \dots, v_n\msetr\) and a vector \(w\) in \(\Z_{2^m}^d\), decide whether there exists a submultiset \(\tilde{V}\subset V\) such that \(\sum_{v\in \tilde{V}} v = w\).
\end{definition}
Note that we can reduce the search version of this problem (finding a suitable \(\tilde{V}\)) to the decision version and the other way around \cite{kabanets16}.
We will now show that the problem instances described in Thm. \ref{thm:main} are equivalent to a certain set of \(d\)\hyp dimensional Subset Sum (\(d\)\hyp SSS) instances. In \Appref{sec:app_proof}, we show that these instances are computationally hard.

\subsection{Equivalence to \(d\)\hyp SSS}
\label{sec:equivalence}
We assume that the server is given the messages \(M=(\msetl a_1,\dots,a_N\msetr,\allowbreak \msetl b_1,\dots,b_{NK}\msetr)\)
(as defined in \Secref{sec:privacy_overview})
, where all vector entries are encoded with \(m\) bits. We assume further that at least two clients are honest, w.l.o.g.\ clients 1 and 2, and that the summand that the server is searching for was sent by either client 1 or client 2. For the moment we will ignore the messages of all other clients. We hence work with the messages \((\msetl a_1,a_2\msetr, \msetl b_1,\dots,b_{2K}\msetr)\), and the task of the adversarial server is to determine whether there exists a \(K\)\hyp element submultiset \(\tilde{V}\subset\msetl b_1,\dots,b_{2K}\msetr\) such that either \(a_1 + \sum_{\tilde{v}\in \tilde{V}} \tilde{v} = h\) or \(a_2 + \sum_{\tilde{v}\in \tilde{V}} \tilde{v} = h\). We can equivalently formulate this as proving or disproving the existence of a submultiset \(\tilde{V}\) such that either \(\sum_{\tilde{v}\in \tilde{V}} \tilde{v} = h - a_1\) or \(\sum_{\tilde{v}\in \tilde{V}} \tilde{v} = h - a_2\). Since the noisy summands \(a_1\) and \(a_2\) are uniformly random, the right sides of these two equations are uniformly random too.
Thus, the server's task is equivalent to solving a \(d\)\hyp SSS problem with a multiset of uniformly random vectors, a uniformly random target sum \(w\), and the additional constraint that the submultiset \(\tilde{V}\) needs to be of cardinality \(K\).

The messages of clients other than 1 or 2 are independent of those of clients 1 and 2 and can therefore be ignored as pure noise. Including them would only increase the chance of false positives, i.e., solutions to the \(d\)\hyp SSS search problem that do not correspond to a set of vectors sent by a single client.

\subsection{Hardness Guarantee}
\label{sec:hardness}
We show that the set of \(d\)\hyp SSS instances from \Secref{sec:equivalence}, which an adversary would have to solve to break privacy, is computationally hard for the parameter choice \(K=dm/2\). In the one\hyp dimensional case and without the additional constraint that the submultiset needs to be of cardinality \(K\), this has been done already by
Impagliazzo and Naor
\cite{impagliazzo1996efficient}.
In the following, we extend their proof to the \(d\)\hyp dimensional setting with a fixed submultiset cardinality.
We first formalize the problem using a similar notation as Impagliazzo and Naor, but invert it: whereas in their case the number of vectors \(K\) is fixed, we fix the encoding length \(m\) and write \(K\) as a function of \(m\):
\begin{definition}
\label{def:hard_problem}
Let \(B=\msetl b_1,\dots,b_{2K(m)}\msetr\) be a multiset of vectors drawn uniformly and independently from \(\Z_{2^m}^d\). The \emph{S5 problem} is the problem of inverting the function
\(g_B(S)=(B,\sum_{b\in S} b)\),
where \(S\) is a uniformly randomly drawn submultiset of \(B\) with cardinality \(K(m)\).
\end{definition}
Note that \(S\) can be represented as a vector \(t\in \{0,1\}^{2K(m)}\) with \(L_1\)\hyp norm equal to \(K(m)\) (\(b_i \in S\) iff \(t_i = 1\)); we will use the representation in the following.
We are interested in hard instances of this problem, depending on the number of vectors \(2K(m)\), i.e., instances for which the function from \Defref{def:hard_problem} is hard to invert. For this we extend the usual definition of one\hyp way functions \cite{goldreich2001foundations} to sequences of functions \(\{g_n\}\), where \(g_n\) is used for inputs of length \(n\) and may be random. In our case, \(g_{n/2K(m)dm}=g_B\) for a multiset \(B\) of \(2K(m)\) \(d\)\hyp dimensional random vectors and an encoding length of \(m\).
\begin{definition}[{\cite{impagliazzo1996efficient}}]
Let \(\{g_n\}\) be a sequence of (potentially random) functions defined on \(D_n\subset \{0,1\}^n\) and let \(g^*:\bigcup_n D_n \rightarrow\{0,1\}^*\) be defined by its restrictions to the \(D_n\):
\(\restr{g^*}{D_n} = g_n\).
\(\{g_n\}\) is \emph{one\hyp way} if the following two conditions hold:
\begin{itemize}
\item \(g^*(t)\) is computable in polynomial time for every \(t\in\bigcup_n D_n\).
\item Let \(\{t_n\}\) be a sequence of uniformly random inputs,
\(t_n{\sim} \mathcal{U}(D_n)\) i.i.d.
For every probabilistic polynomial\hyp time algorithm \(A\) (that attempts to invert \(g^*\)) and for all \(c>0\),
\begin{equation*}
\Pr(g^*(A(g^*(t_n))) = g^*(t_n)) < n^{-c}
\end{equation*}
for all sufficiently large \(n\).
\end{itemize}
\end{definition}
Since solving the search version of the S5 problem is exactly the problem of inverting a function \(g_n\), we call sets of instances for which the corresponding sequence \(\{g_n\}\) is one\hyp way \emph{hard} \cite{impagliazzo1996efficient}. In our case, sets of instances are defined by the number of messages as a function of the encoding length \(2K(m)\).
Using this definition, the hardest instances are those for which \(2K(m)=dm\):
\begin{theorem}[\cf\ {\cite[Prop.~1.2]{impagliazzo1996efficient}}]
\label{thm:hardness}
\leavevmode
\begin{enumerate}
\item Let \(2K'(m) \leq 2K(m) \leq dm\). If the S5 problem is hard for \(K'(m)\), then it is also hard for \(K(m)\).
\item Let \(dm \leq 2K(m) \leq 2K'(m)\). If the S5 problem is hard for \(K'(m)\), then it is also hard for \(K(m)\).
\end{enumerate}
\end{theorem}
For the proof of \Thmref{thm:hardness} we need to characterize the instances of the S5 problem in the two different cases. In the first case the function \(g_B\) is almost injective, while in the second case it is almost surjective with all values in its range occurring almost the same number of times.
\begin{definition}[\cite{impagliazzo1996efficient,impagliazzo1989recycle}]
Let \(D\) be a probability distribution on \(\{0,1\}^n\). We say \(D\) is \emph{quasi\hyp random} within \(\varepsilon\), if for all \(U\subset \{0,1\}^n\) we have that \(\abs{\Pr_D(U) - \abs{U}/2^n} < \varepsilon\).
\end{definition}
\begin{restatable}[\cf\ {\cite[Prop.1.1]{impagliazzo1996efficient}}]{lemma}{hardinstances}
\label{thm:characterization}
\leavevmode
\begin{enumerate}
\item Let \(2K(m)\leq cdm\) for \(c<1\). Let \(B\) and \(S\) both be chosen uniformly at random. Except with probability exponentially small, there is no \(S'\neq S\) such that \(g_B(S)=g_B(S')\).
\item Let \(2K(m) \geq cdm\) for \(c>1\). Let \(B\) be chosen uniformly at random. Except with probability exponentially small (w.r.t. the choice of \(B\)), the distribution given by \(g_B(S)\) for a randomly chosen \(S\) is quasi-random (w.r.t.\ \(S\)) within an exponentially small amount.
\end{enumerate}
\end{restatable}
The proof of Lemma~\ref{thm:characterization} can be found in \Appref{sec:app_proof}.

\Thmref{thm:hardness} can now easily be deduced from \Lemmaref{thm:characterization} (see \cite{impagliazzo1996efficient}). Assume we have an algorithm that efficiently solves the S5 problem for instances with \(2K(m)\) vectors where the function \(K\) is chosen such that \(g_B\) is almost injective. Instances with \(2K'(m) < 2K(m)\) vectors can be transformed into instances with \(2K(m)\) vectors by removing enough of the least significant bits. This adds only few false positives, i.e., solutions to the inversion of \(g_B\), due to instances with \(2K(m)\) vectors being almost injective. Similarly, if we have an algorithm for instances with \(2K(m)\) vectors where \(g_B\) is almost uniform, we transform instances with \(2K'(m) > 2K(m)\) vectors by adding random bits at the end. Every solution to the modified
inversion
problem is a solution to the original
inversion
problem, and we do not lose many solutions because the sums \(g_B(S)\) are almost uniformly distributed for \(2K(m)\) vectors.


\subsection{Hardness in Practice}
\label{sec:hardness_practice}
One\hyp dimensional SSS is an NP hard problem \cite{Karp1972}, its multi\hyp dimensional generalization therefore is as well.
However, multi\hyp dimensional SSS has been mostly neglected by the research community so far, apart from a negative result about its approximability \cite{emiris2017approximating}.
One\hyp dimensional subset sum, on the other hand, has a long history of study. Similar to this paper, instances are typically characterized in terms of the ratio of the number of messages \(n\) and their encoding length \(l(n)\), where the optimal choice for security is \(l(n) = n\) \cite{impagliazzo1996efficient}. For \(l(n) > 1.06n\), SSS can be transformed into a lattice shortest vector problem \cite{joux1991improving,coster1991improved} that can be solved efficiently for certain instances but is, like SSS, NP hard in the general case. For \(l(n)=\mathcal{O}(\log(n))\), there exists a very efficient dynamic programming solution \cite{galil1991almost}. Instances with \(l(n)=n\) are hard instances in the same sense as in our paper: The number of possible summands equals the number of bits per summand. The fastest algorithms that solve such instances still require exponential time. The fastest traditional algorithm runs in \(\tilde{\mathcal{O}}(2^{0.291n})\) \cite{becker2011improved}, the fastest quantum algorithm in time \(\tilde{\mathcal{O}}(2^{0.226n})\) \cite{bernstein2013quantum}, where the notation \(\tilde{\mathcal{O}}\) suppresses polynomial factors. Thus, despite significant efforts, no efficient algorithm has been found for the hardest set of one\hyp dimensional SSS instances and it seems likely that this will be also the case for its multi\hyp dimensional counterpart.
Note that for small \(n\), the problem is still solvable. This translates to both \(d\) and \(m\) being small. Since \(m=\lceil\log(N)\,\tilde{m}\rceil\) and the server could lie to the clients about \(N\), in the case of \(d\) being small, \(\tilde{m}\) has to be chosen sufficiently large. Because this choice is transparent to the clients, they are able to detect when the server chooses a too small value and can refuse to participate in the training.

\section{Experiments}
\label{sec:experiments}
To assess the feasibility of our method, we conduct experiments with various parameter settings, where we measure the wall clock runtime and the amount of transmitted data. Experiments assessing the accuracy of the computed sum are not necessary, since the sum is computed exactly by design.

\subsection{Methodology}
We simulate the protocol on a single machine by routing all communication over the loopback interface. This means that the runtimes include the overhead of the network protocols, but not the delay of the network connection. The machine is equipped with an Intel Xeon E5-2680 v3, 256 GB of RAM, and running Ubuntu 18.04 LTS. We restrict the server to a single CPU core. This is similar to the experimental setting used by
Bonawitz et al.
\cite{bonawitz2017practical}. To have a reference value, we compare our method with sending the summands directly without any privacy measures.
We test with \(m=16\), \(m=32\) and \(m=64\) bits, and with summand vectors of dimension 10, 100 and 1000.
We sample the summands \(s_i\) uniformly at random from the admissible space. As in \Secref{sec:cost_analysis}, the parameter vector \(p\) of the query function is sampled uniformly at random from the same space.
We perform experiments with 8 and with 128 clients. Due to the \((m-\tilde{m})\) bits that need to be reserved for the carry in the summation, the number of bits per summand reduces by 3 in the former and 7 in the latter case.
In all experiments, we compute the sum of the \(s_i\) directly and compare it with the sum computed by S5 to ensure correctness.
The implementation is written in Python and based on NumPy, Flask and Gunicorn.\footnote{The code for generating the data, running the experiments and producing the plots is available on GitHub: \url{https://github.com/valentin-hartmann-research/S5}.}

\subsection{Results}
\xhdr{Runtime}
As we can see in \Figref{fig:runtime}, the runtime increases almost linearly with the number of dimensions \(d\) and the number of bits \(m\). At first, this seems to contradict the asymptotic results in \Secref{sec:cost_analysis}, which indicate an increase of the runtime that is quadratic in both quantities. However, in an implementation, the computational cost for managing the network communication far outweighs the cost for generating the random vectors from the seeds, at least for our parameter settings. For very large values of \(d\) and \(m\) we would see the quadratic behavior.

Even for the fairly high number of 128 clients and 1000 dimensions, the summation of 25\hyp bit numbers (32 bits minus the \(\log(128)\) bits reserved as carry bits) takes just above half an hour, making the method feasible even for gradient descent\hyp based training of machine learning models that requires many iterations.
The non\hyp private baseline is able to perform the summation for all parameter settings in at most 1.5 seconds.

\xhdr{Transmitted data}
We measure both the data sent and received by the clients. The results are shown in \Figref{fig:traffic}. As predicted by the theory, neither value changes with the number of clients, hence we only show one set of plots for both 8 and 128 clients.
The behavior w.r.t.\ \(d\) and \(m\) is almost linear and thus very similar to the one of the runtime. The asymptotics in \Secref{sec:cost_analysis} suggest an additional logarithmic factor. This comes from the fact that for more dimensions and more bits the number of bits for the seeds needs to increase. But for our parameter settings, the size of the seeds is far outweighed by the overhead of the network protocols, which makes this factor disappear.

Again looking at the case of 1000 dimensions and 32 bits, only less than 25 MB need to be transmitted by each client, making the method suitable even for settings with slow internet connections.
In the case of the non\hyp private baseline, less than 1.5 MB of data need to be transmitted in any experiment.

\begin{figure}
    \centering
    \subfloat[8 clients, hence 13, 29 and 61 bits per summand, respectively.]{\includegraphics[width=.49\linewidth]{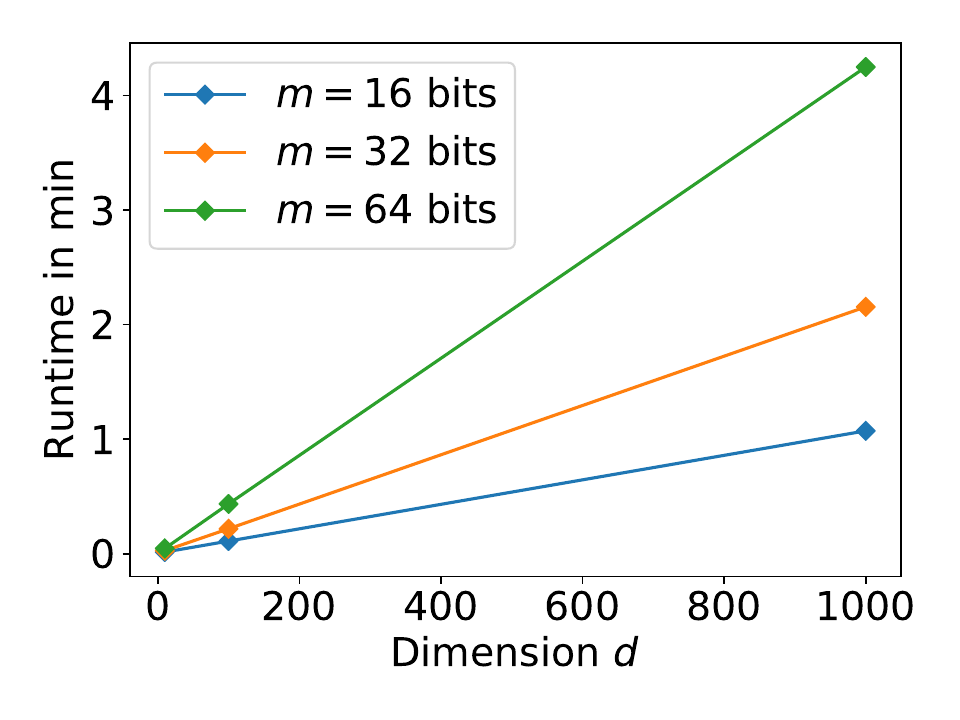}}
    \hspace*{\fill}
    \subfloat[128 clients, hence 9, 25 and 57 bits per summand, respectively.]{\includegraphics[width=.49\linewidth]{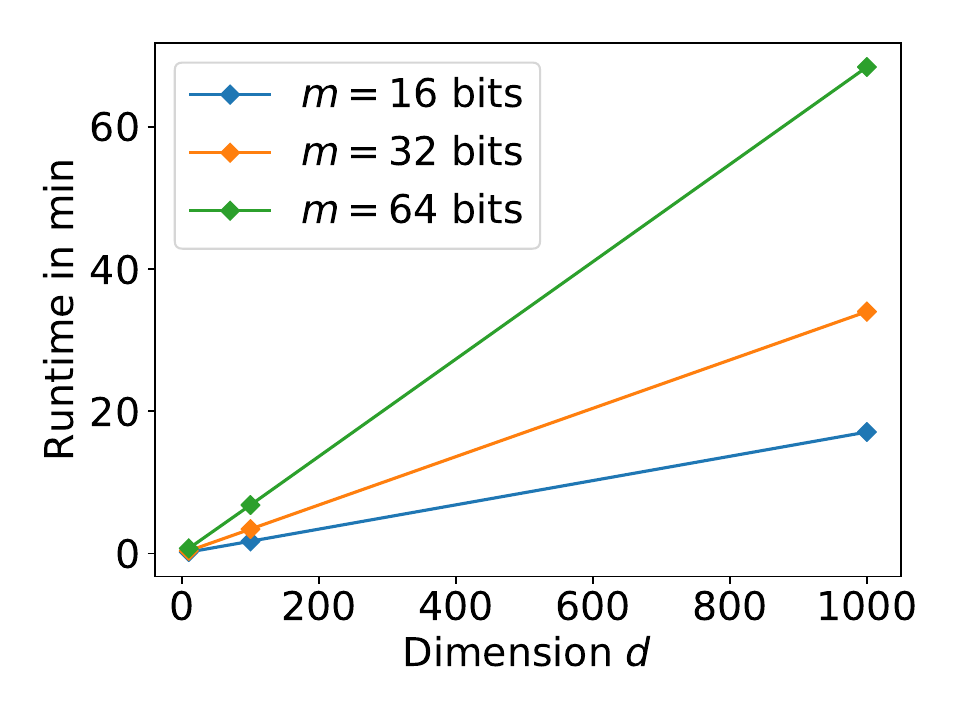}}
    
    \caption{Wall clock time for the summation. Measurements were done for
    \(d=10,\ 100,\ 1000\).}
    \label{fig:runtime}
\end{figure}

\begin{figure}
    \centering
    \subfloat[Amount of data sent from each client to the server.]{\includegraphics[width=.49\linewidth]{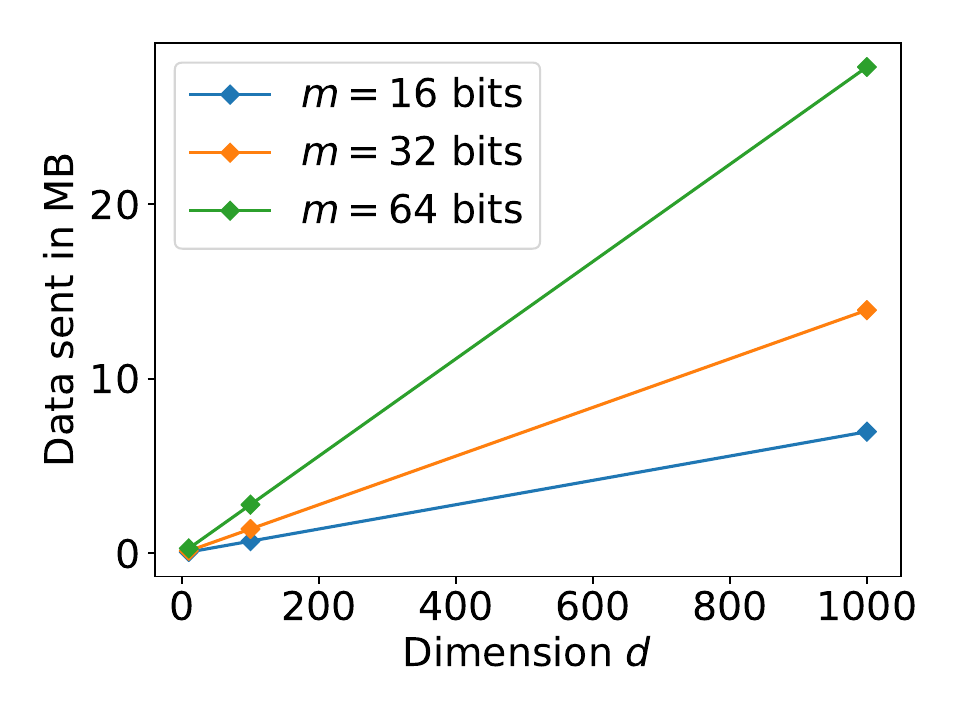}}
    \hspace*{\fill}
    \subfloat[Amount of data sent from the server to each client.]{\includegraphics[width=.49\linewidth]{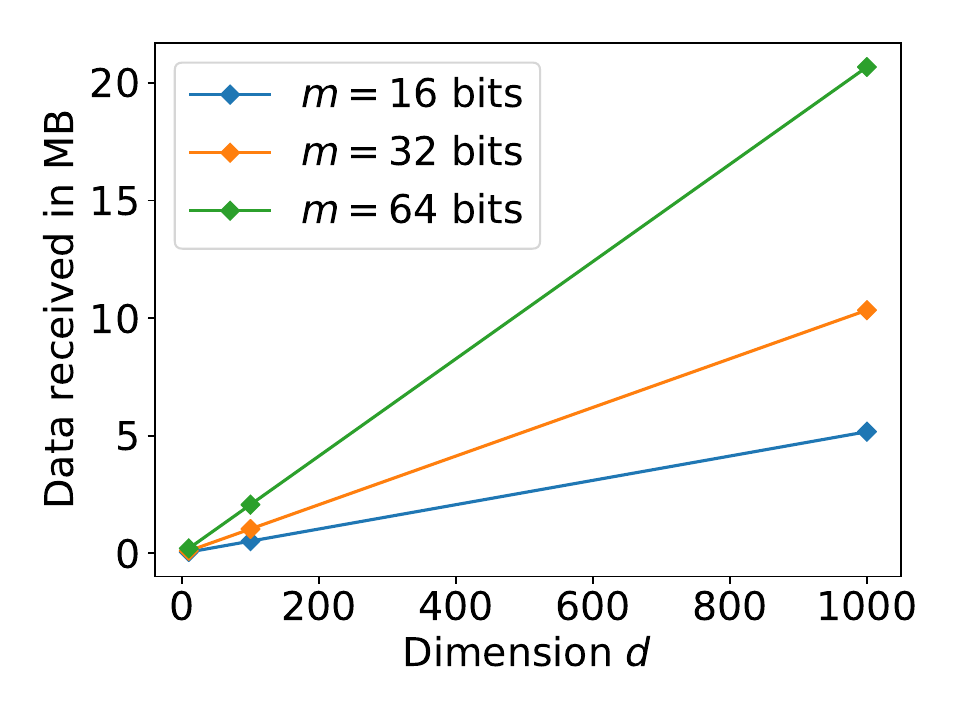}}
    
    \caption{Data transmitted by a single client. Measurements were done for
    \(d=10,\ 100,\ 1000\).}
    \label{fig:traffic}
\end{figure}

\section{S5 and Federated Learning}
S5 can be combined with federated learning (FL) \cite{federatedlearning} to provide privacy when training machine learning models in a distributed setting. FL is designed for settings where different clients each hold a dataset, and a machine learning model is to be trained on the union of all those datasets via stochastic gradient descent (SGD). In each iteration, the server ships the current model to all clients. The clients then compute the model's gradient w.r.t. their local data and share this gradient with the server. The server averages the received gradients and performs a gradient step to update the model parameters. The averaging step can be implemented via S5 to prevent the extraction of private information from the gradients. In our experiments in \Secref{sec:experiments} we showed that gradients with a thousand dimensions can be summed in less than an hour. For models with \(d\gg 1000\) dimensions, gradient sparsification or compression techniques can be used. In our setting, we can, for example, use random\hyp \(k\) sparsification: in each iteration, clients only send the values of (the same) \(k\ll d\) random entries of their local gradient instead of the full gradient to the server. This scheme has been shown to have the same rate of convergence as vanilla SGD \cite{stich2018sparsified}.
\label{sec:fl}

\section{Limitations}
\label{sec:limitations}
In S5, messages are routed through an anonymization network. We use this network as a building block and are therefore not concerned with certainly existing weaknesses of specific implementations such as Tor. Further, making as few assumptions as we do does not come without a cost. If a client drops out during the summation, the protocol has to be restarted, since the incomplete sum that the server receives is uniformly randomly distributed on the entire range \(\Z_{2^m}^d\). Hence, S5 is best suited for settings where clients have stable connections to the server.
We leave finding ways to relax this restriction to future work, as well as ensuring correctness, i.e., preventing clients from poisoning the summation by sending forged summands.

\section{Conclusion}
\label{sec:conclusion}
In this paper, we developed S5, a protocol for computing sums of distributed vectors in a privacy\hyp preserving fashion. It only requires two honest clients, no peer\hyp to\hyp peer connections between clients, and does not distort the sum. This is a big improvement over prior work on computational privacy guarantees that typically either assumes a trusted server or a larger number of trusted clients.
In experiments, we have shown that even with more than 100 clients, S5 can compute sums of 1000\hyp dimensional vectors in less than an hour, making it useful for applications such as Federated Learning; when combined with gradient sparsification, even for high\hyp dimensional models.
To the best of our knowledge, this paper is also one of the first ones to provide an analysis of the multidimensional subset sum problem, which we believe to be of independent interest, both from a theoretical perspective and for its use in other privacy and cryptography applications.

\bibliographystyle{splncs04}
\bibliography{references}

\appendix

\section{Overview of S5}
\Figref{fig:protocol} provides an overview of S5.
\begin{figure*}
\centering
\procedure{S5}{
\textbf{Server} \> \> \textbf{Client \(i\)} \\
\rightmsgsmall{\text{Send query parameters } p} \\
\> \> \text{Compute summand } s_i \leftarrow f_p(D_i) \\
\> \> \text{Sample seeds } R_{i1},\dots,R_{iK} \\
\> \> r_{ik}\leftarrow\RNG(R_{ik}),\dots,r_{iK}\leftarrow\RNG(R_{iK}) \\
\leftmsgsmall{\text{Send } \tilde{s}_i = s_i + \sum_{k=1}^K r_{ik}} \\
\leftmsgsmall{\text{Send } R_{i1},\dots,R_{iK}} \\
r_{11}\leftarrow\RNG(R_{11}),\dots,r_{NK}\leftarrow\RNG(R_{NK}) \> \> \\
\text{Compute sum } s\leftarrow\sum_{i=1}^N \left(\tilde{s}_i - \sum_{k=1}^K r_{ik}\right) \> \>
}
\caption{Overview of our method. The messages from the clients are sent through a mixnet such as Tor.}
\label{fig:protocol}
\end{figure*}

\section{Preprocessing}
\label{sec:app_preprocessing}
To transform real\hyp valued summands into vectors with entries from \(\Z_{2^m}\), there are two things that we need to do: (1) upper bound the summands and (2) represent them as non\hyp negative integers.

We first choose a number of bits \(\tilde{m}\) to use for the representation of each entry, and the position of the decimal point, i.e., how many bits are used for the integer and how many for the fractional part. In the following we assume for simplicity and w.l.o.g.\ that all bits are used for the integer part. In the case of gradient descent, there might not exist an a priori bound on the gradient entries, which naturally determines the number of integer bits. In this case, we can choose a large value for \(\tilde{m}\), but might still encounter gradients \(s_i\) with one or more entries with absolute value larger or equal \(2^{\tilde{m}-1}\). We project those gradients to the \(L_\infty\) ball with radius \(2^{\tilde{m}-1} - 1/2\) around 0: \(s_i \leftarrow (2^{\tilde{m}-1} - 1/2) s_i/\norm{s_i}_\infty\). This operation, known as clipping, preserves the ratio of the gradient entries w.r.t. each other and is commonly used in differentially private machine learning \cite{abadi2016deep} and also non\hyp differentially private neural network training \cite{pascanu2013difficulty}.

To obtain non\hyp negative vectors, we replace each \(s_i\) by \(s_i + (2^{\tilde{m}-1} - 1/2) \mathbbm{1}\), where \(\mathbbm{1}\) denotes the \(d\)\hyp dimensional 1\hyp vector. To reverse this transformation, the server can simply subtract \(N (2^{\tilde{m}-1} - 1/2)\) from the final \(s\).

Now all \(s_i\) lie in \([0,2^{\tilde{m}} - 1]^d\). We still need to discretize them to integers. This can be done by, e.g., rounding the entries stochastically to the nearest integer, i.e., for an integer \(i\) such that \(i \leq t \leq i+1\), we round \(t\) to \(i\) with probability \(t-i\) and to \(i+1\) with probability \(i+1-t\). In expectation, this does not change the value of \(t\). In the beginning we mentioned that we would work in \(\Z_{2^m}\). We choose \(m \geq \tilde{m} + \log(N)\) so that the sum over all \(s_i\) does not exceed \(2^m\).

We would like to remark that for gradient descent the discretization of the gradients comes with only a small decrease in model performance \cite{gupta2015deep}.

\section{Proof of Lemma 7}
\label{sec:app_proof}
We prove \Lemmaref{thm:characterization} from \Secref{sec:hardness}.

\hardinstances*
\begin{proof}
The proof closely follows that of Prop.~1.1 from Impagliazzo and Naor \cite{impagliazzo1996efficient}.
\begin{enumerate}
    \item Let \(S'\neq S\) be a different \(K(m)\)\hyp element subsets of the \(2K(m)\)\hyp element set \(B\) in \(d\) dimensions with \(2^m\) bits per dimension. Then \(g_B(S)\) and \(g_B(S')\) are independent uniformly random vectors, as mentioned earlier. There are \(\binom{2K(m)}{K(m)} - 1 \leq 2^{2K(m)}\) possible choices for \(S'\). Hence we have
    \begin{align*}
        &\Pr(\exists S'\neq S: g_B(S) = g_B(S'))\\
        &\leq \sum_{S'\neq S} \Pr(g_B(S) = g_B(S'))\\
        &\leq 2^{2K(m)} 2^{-dm} \leq 2^{-(1-c)dm}.
    \end{align*}
    \item Because \(g_B(t)\) and \(g_B(t')\) are independent and uniformly random w.r.t.\ the choice of \(B\) for \(t\neq t'\), \(\{g_B\}_{B\subset \Z_{2^m}^d, \abs{B}=2K(m)}\) is a family of universal hash functions from \(\{0,1\}^{2K(m)}\) to \(Z_{2^m}^d\). We can therefore apply the following lemma from Santha and Vazirani \cite{santha1984generating}:
    \begin{lemma}[Leftover Hash Lemma \cite{santha1984generating}]
    \label{thm:leftover}
    Let \(U\subset \{0,1\}^n\), \(\abs{U} \geq 2^l\). Let \(e>0\) and let \(G\) be an almost universal family of hash functions mapping \(n\) bits to \(l-2e\) bits. Then the distribution \((g,g(u))\) is quasi\hyp random within \(1/2^e\) (on the set \(G\times \{0,1\}^{l-2e}\)), where \(g\) is chosen uniformly at random from \(g\), and \(u\) uniformly from \(U\).
    \end{lemma}
    Since we only allow subset sums where exactly half of the vectors is summed up, the domain of our hash functions is restricted to \(U=\{v\in \{0,1\}^{2K(m)}:\ \norm{v}_1 = K(m)\}\). This set has a cardinality of \(\binom{2K(m)}{K(m)} \geq 2^{2K(m)}/(2K(m)+1) = 2^{2K(m) - \log(2K(m) + 1)}\), whereas the domain of \(g_B\) has cardinality \(2^{dm}\). We thus get \(e=2K(m) - \log(2K(m) + 1) - dm \leq (c-1-\mathcal{O}(log(dm)/dm))dm\) for the \(e\) from Lemma~\ref{thm:leftover}, which yields, for all \(T\subset \{0,1\}^{dm}\):
    \begin{align*}
        \mathop{\mathbb{E}}_B\abs{\Pr_{u}(g_B(u)\in T) - \frac{\abs{T}}{2^{dm}}} < 2^{-(c-1-\mathcal{O}(\frac{log(dm)}{dm}))dm}.
    \end{align*}
    Because this bound on the expectation w.r.t.\ \(B\) is exponential, Markov's inequality asserts that an exponential bound holds for all but an exponentially small fraction of all \(B\):
    \begin{align*}
        &\Pr_B(\abs{\Pr_{u}(g_B(u)\in T) - \frac{\abs{T}}{2^{dm}}} \geq 2^{-\frac{c-1}{2}dm})\\
        &\leq 2^{-(\frac{c-1}{2}-\mathcal{O}(\frac{log(dm)}{dm}))dm}.
    \end{align*}
\end{enumerate}
\end{proof}

\end{document}